\documentclass[aps,pra,twocolumn,superscriptaddress,floatfix,nofootinbib,showpacs,longbibliography,groupedaddress]{revtex4-1}

\usepackage[utf8]{inputenc}  
\usepackage[T1]{fontenc}     
\usepackage[british]{babel}  
\usepackage[sc,osf]{mathpazo}
\usepackage[scaled=0.86]{berasans}  
\usepackage[colorlinks=true, citecolor=blue, urlcolor=blue]{hyperref}  
\usepackage{graphicx} 
\usepackage[babel]{microtype}  
\usepackage{amsmath,amssymb,amsthm,bm,amsfonts,mathrsfs,bbm} 

\usepackage{xspace}  
\usepackage{pgf,tikz}
\usepackage{xcolor}
\usepackage{multirow}
\usepackage{array}
\usepackage{bigstrut}
\usepackage{braket}
\usepackage{color}
\usepackage{natbib}
\usepackage{multirow}
\usepackage{mathtools}
\usepackage{float}
\usepackage[caption = false]{subfig}
\usepackage{xcolor,colortbl}
\usepackage{color}

\newcommand{\be}{\begin{equation}}
\newcommand{\ee}{\end{equation}}
\newcommand{\ba}{\begin{eqnarray}}
\newcommand{\ea}{\end{eqnarray}}

\newtheorem{theorem}{Theorem}

\newtheorem{proposition}{Proposition}

\newtheorem{remark}{Remark}

\def\>{\rangle}
\def\<{\langle}

\usepackage{centernot}
\usepackage{subfig}

\begin{document}

\title{On the state space structure of tripartite quantum systems}

\author{Hari krishnan S V}
\thanks{These authors have contributed equally.}   
\affiliation{School of Physics, IISER Thiruvanathapuram, Vithura, Kerala 695551, India.}

\author{Ashish Ranjan}
\thanks{These authors have contributed equally.}    
\affiliation{School of Physics, IISER Thiruvanathapuram, Vithura, Kerala 695551, India.}

\author{Manik Banik}
\affiliation{School of Physics, IISER Thiruvanathapuram, Vithura, Kerala 695551, India.}

\begin{abstract}
State space structure of tripartite quantum systems is analyzed. In particular, it has been shown that the set of states separable across all the three bipartitions [say $\mathcal{B}^{int}(ABC)$] is a strict subset of the set of states having positive partial transposition (PPT) across the three bipartite cuts [say $\mathcal{P}^{int}(ABC)$] for all the tripartite Hilbert spaces $\mathbb{C}_A^{d_1}\otimes\mathbb{C}_B^{d_2}\otimes\mathbb{C}_C^{d_3}$ with $\min\{d_1,d_2,d_3\}\ge2$. The claim is proved by constructing state belonging to the set $\mathcal{P}^{int}(ABC)$ but not belonging to $\mathcal{B}^{int}(ABC)$. For $(\mathbb{C}^{d})^{\otimes3}$ with $d\ge3$, the construction follows from specific type of multipartite unextendible product bases. However, such a construction is not possible for $(\mathbb{C}^{2})^{\otimes3}$ since for any $n$ the bipartite system $\mathbb{C}^2\otimes\mathbb{C}^n$ cannot have any unextendible product bases [\href{https://doi.org/10.1103/PhysRevLett.82.5385}{Phys. Rev. Lett. {\bf 82}, 5385 (1999)}]. For the $3$-qubit system we, therefore, come up with a different construction.    
\end{abstract}

\maketitle
 
\section{Introduction}
Hilbert space quantum mechanics provides an extremely precise mathematical description for microscopic phenomena. It associates tensor product Hilbert spaces with composite quantum systems and results in entangled quantum states which have no analog in classical physics \cite{EPR35,Schrdinger1935,Schrdinger1936}. The advent of quantum information theory \cite{Chuang00} identifies several important applications of entanglement (see \cite{Ghne2009,Horodecki2009} and references therein). Characterization, identification, and quantification of quantum entanglement are thus questions of great practical interest. One of the most widely used tests for bipartite states' entanglement verification is the positive partial transposition (PPT) criterion. States having negative partial transposition (NPT) are always entangled \cite{Peres96}, whereas PPT implies separability only for the systems with composite dimensions no more than six \cite{Horodecki1996}.  In other words, being PPT is a necessary and sufficient condition for separability for the composite systems with dimensions $\leq 6$. Consequently, for higher dimensional systems one has several hierarchical convex-compact subsets of states within the set of allowed quantum states. Identifying these subsets as well as their boundaries is essential to understand the intricacy of quantum state space structure and entanglement properties of the quantum states. Having the convex-compact structures, these sets allow the classic Minkowski-Hahn–Banach separation theorem to characterize their several important features \cite{Boyd04}.    

The complexity of the situation increases rapidly with an increase in the number of component subsystems comprising the composite system \cite{Eggeling01}. For instance, separability/ PPT-ness can be considered across different bipartite cuts and accordingly one ends up with different convex-compact subsets of states. Several intriguing structures consequently emerge. For instance, a $3$-qubit state may not be fully separable even if it is separable across every possible bipartite cut and hence contains multipartite entanglement \cite{Bennett99}. On the other hand, for such a system the set of biseparable states is strictly contained within the PPT mixture set \cite{Ha16}. In this work,  we show that the set of states that are separable across all possible bipartitions are strictly contained within the set of states that are PPT across all possible bipartite cuts. We prove this claim for an arbitrary tripartite Hilbert space $\mathbb{C}^{d_1}\otimes\mathbb{C}^{d_2}\otimes\mathbb{C}^{d_3}$ with $\min\{d_1,d_2,d_3\}\ge2$. The proof is constructive. We construct states that are PPT across all three bipartitions but inseparable across some bipartite cuts. Construction for $(\mathbb{C}^{d})^{\otimes3}$ with $d\ge3$ follows from a specific kind of unextendible product bases. However, such a construction is not possible for $3$-qubit Hilbert space as there cannot be any set of mutually orthogonal product states in $\mathbb{C}^2\otimes\mathbb{C}^n$ that are not completable \cite{Bennett99(1),DiVincenzo03}. To prove the claim for $3$-qubit system we, therefore, come up with a different construction.  

The manuscript is organized as follows: in section \ref{sec2}, we note down the notations and recall some preliminary results; all our findings are listed in section \ref{sec3}; finally in section \ref{sec4} we put our conclusions along with some open questions for further research.

\section{Notations and Preliminaries}\label{sec2}
A quantum system $S$ is associated with a complex separable Hilbert space. $\mathcal{H}_S$ over complex field \cite{Kraus83,Chuang00}. Throughout this work, we will consider only finite dimensional systems and hence $\mathcal{H}_S$ will be isomorphic to some complex space $\mathbb{C}^d_S$, with $d$ being the dimension of the complex vector space. The system's state is described by a density operator $\rho_S$ (positive operator with unit trace) acting on $\mathbb{C}^d_S$. Collection of density operators form a convex compact set $\mathcal{D}(S)$ embedded in $\mathbb{R}^{d^2-1}$. We will also sometimes specify this set as $\mathcal{D}(d)$ to distinguish between systems with different dimensions.  The extreme points of $\mathcal{D}(S)$ are called pure states, and they satisfy $\rho^2_S=\rho_S$. Let $\mathcal{E}_{\mathcal{D}}(S)$ denotes the set of all extremal points of $\mathcal{D}(S)$. Such an extremal state can also be considered as a rank$-1$ projector, {\it i.e.} $\rho_S=\ket{\psi}_S\bra{\psi}$ for some $\ket{\psi}_S\in\mathbb{C}^d_S$ whenever $\rho_S\in\mathcal{E}_{\mathcal{D}}(S)$. States that are not pure are called mixed states and they allow convex decomposition in terms of pure states, {\it i.e.} $\forall~\rho_S\not\in\mathcal{E}_{\mathcal{D}}(S),~\exists$ $\sigma^i_S\in\mathcal{E}_{\mathcal{D}}(S)$ s.t $\rho_S=\sum p_i\sigma^i_S$, where $p_i>0~\&~\sum p_i=1$.

Hilbert space $\mathbb{C}_{AB\cdots}$ of a composite system consisting of component subsystems $A,B,\cdots$ is given by the tensor product of the Hilbert spaces associated with the component subsystems, {\it i.e.} $\mathbb{C}^d_{AB\cdots}=\mathbb{C}^{d_A}_{A}\otimes\mathbb{C}^{d_B}_{B}\otimes\cdots$. Here $d_A,d_B,\cdots$ denotes the dimension of the component subsystems while the dimension of the composite system is $d=d_A\times d_B\times\cdots$. While the axiomatic formulations of quantum mechanics contain this tensor product postulate \cite{Neumann55,Dirac66,Jauch68}, recent developments indicate that this assumption can be logically derived from the state postulate and the measurement postulate \cite{Carcassi21}. 

A bipartite state $\rho_{AB}\in\mathcal{D}(AB)$ is called a pure product state if and only if $\rho_{AB}=\ket{\chi}_{AB}\bra{\chi}$, where $\ket{\chi}_{AB}=\ket{\psi}_A\otimes\ket{\phi}_B$ for some $\ket{\psi}_A\in\mathbb{C}_A~\&~\ket{\phi}_B\in\mathbb{C}_B$. The convex hull of these product states will be denoted as $\mathcal{S}(AB)$, and the states in $\mathcal{S}(AB)$ are generally called separable states. States that are not separable are called entangled, {\it i.e.} an entangled state $\rho^{ent}_{AB}\in\mathcal{D}(AB)$ but $\rho^{ent}_{AB}\not\in\mathcal{S}(AB)$. While $\mathcal{E}_{\mathcal{D}}(AB)$ consists of pure product states and pure entangled states, $\mathcal{E}_{\mathcal{S}}(AB)$ consists of only pure product states. For a bipartite system one can define another convex compact set, namely the Peres set $\mathcal{P}(AB)$: the set of states with positive-partial-transpose (PPT) \cite{Peres96}. A bipartite state $\rho_{AB}$ belongs to the set $\mathcal{P}(AB)$ whenever $\rho_{AB}^{\mathrm{T}_k}\ge0$, where $T_k$ denotes (partial) transposition with respect to the $k^{th}$ subsystem with $k\in\{A,B\}$. The following set inclusion relations are immediate:
\begin{align}\label{eq1}
\mathcal{S}(AB)\subseteq\mathcal{P}(AB)\subsetneq\mathcal{D}(AB).
\end{align}
Equality between the first two sets holds only for the composite system of dimension no more than $6$ \cite{Horodecki1996}. For the higher dimensions $\mathcal{S}(AB)$ is known to be a proper subset of $\mathcal{P}(AB)$ \cite{Horodecki97,Bennett99,Brub00,DiVincenzo03}. A bipartite state $\rho_{AB}^{ppt}$ is PPT-entangled whenever $\rho_{AB}^{ppt}\in\mathcal{P}(AB)$ but $\rho_{AB}^{ppt}\not\in\mathcal{S}(AB)$. The PPT entangled states exhibit an intriguing irreversible feature: their preparation under local quantum operations assisted with classical communications (LOCC) requires non-zero amount of maximally entangled state to be shared between the subsystems \cite{Horodecki98}, but no maximally entangled state can be distilled from them under LOCC \cite{Vidal01,Horodecki05}. Despite being undistillable, PPT entangled states find several applications, such as activating entanglement distillation for other entangled states \cite{Horodecki99,Vollbrecht02}, multipartite entanglement manipulation  \cite{Ishizaka05}, information processing \cite{Masanes06}, 
private key distillation \cite{Oppenheim05,Horodecki08,Horodecki09}, and quantum metrology \cite{Czekaj15}. Along with the set inclusion relations in Eq.(\ref{eq1}) it also immediately follows that,       
\begin{align}
\mathcal{E}_{\mathcal{S}}(AB)\subseteq\mathcal{E}_{\mathcal{P}}(AB)\subsetneq\mathcal{E}_{\mathcal{D}}(AB).  
\end{align}
It is important to note that a state belonging either in $\mathcal{E}_{\mathcal{S}}(AB)$ or in $\mathcal{E}_{\mathcal{D}}(AB)$ must be a pure state. However, for systems with composite dimension more than $6$ the set $\mathcal{E}_{\mathcal{S}}(AB)$ is strictly contained within $\mathcal{E}_{\mathcal{P}}(AB)$ and contains $\mathcal{E}_{\mathcal{S}}(AB)$ as a proper subset. Therefore, for such higher dimensional systems $\mathcal{E}_{\mathcal{P}}(AB)$ must contain some {\it nontrivial} mixed states along with the pure product states belonging in $\mathcal{E}_{\mathcal{S}}(AB)$. For instance, the PPT entangled state of $\mathbb{C}^3\otimes\mathbb{C}^3$ as constructed in \cite{Bennett99} lies in $\mathcal{E}_{\mathcal{P}}(3\otimes3)$ \cite{Chen2011}. Examples of such extremal states for $\mathcal{E}_{\mathcal{P}}(n\otimes n)$ with odd $n~\&~n\ge3$ can be found in \cite{Halder19}. There are in-fact some efficient methods in literature to check extremality of $\mathcal{P}(AB)$ \cite{Leinaas07,Leinaas10,Leinaas10(1),Augusiak10}.   

Moving to the tripartite system, a pure state $\ket{\chi}_{ABC}$ is called fully product if it is of the form $\ket{\chi}_{ABC}=\ket{\psi}_A\otimes\ket{\phi}_B\otimes\ket{\eta}_C$ for some $\ket{\psi}_A\in\mathbb{C}^{d_A}_A,~\ket{\phi}_B\in\mathbb{C}^{d_B}_B ~\&~\ket{\eta}_C\in\mathbb{C}^{d_C}_C$. The convex hull of pure product states will be denoted as $\mathcal{F}(ABC)$, and a state belonging to this set is generally called fully separable state. A tripartite state is called biseparable across A|BC cut if it of the form $\ket{\chi}_{ABC}=\ket{\psi}_A\otimes\ket{\phi}_{BC}$, where the state $\ket{\phi}_{BC}$ is allowed to be entangled across B|C cut. Convex hull of the states biseparable across A|BC cut will be denoted as $\mathcal{B}(A|BC)$. Similarly, one can define the sets $\mathcal{B}(B|CA)$ and $\mathcal{B}(C|AB)$ that are biseparable across B|CA cut and C|AB cut, respectively. A state belonging in the convex hull of the sets $\mathcal{B}(A|BC),~\mathcal{B}(B|CA) ~\&~\mathcal{B}(C|AB)$ is generally called biseparable, and we denote this set as $\mathcal{B}^{ch}(ABC)$, {\it i.e.}
\begin{align}
\mathcal{B}^{ch}(ABC):=\mbox{Convex~Hull}&\left\{\mathcal{B}(A|BC),\mathcal{B}(B|CA),\right.\nonumber\\
&\left.~~~~~~~\mathcal{B}(C|AB)\right\}. 
\end{align}
We can also consider the intersection of these three sets, which we will denote as $\mathcal{B}^{int}(ABC)$, {\it i.e.}
\begin{align}
\mathcal{B}^{int}(ABC):=\mathcal{B}(A|BC)\cap\mathcal{B}(B|CA)\cap\mathcal{B}(C|AB).    
\end{align}
Analogously, one can define the convex sets of PPT states across a given cut and then can consider the convex hull as well as the intersection of different such sets. We will respectively denote them as $\mathcal{P}(A|BC),~\mathcal{P}(B|CA),~\mathcal{P}(C|AB),~\mathcal{P}^{ch}(ABC)$, and $\mathcal{P}^{int}(ABC)$, where
\begin{subequations}
\begin{align}
\mathcal{P}^{ch}(ABC)&:=\mbox{Convex~Hull}\left\{\mathcal{P}(A|BC),\mathcal{P}(B|CA),\right.\nonumber\\
&\left.~~~~~~~\mathcal{P}(C|AB)\right\};\\
\mathcal{P}^{int}(ABC)&:=\mathcal{P}(A|BC)\cap\mathcal{P}(B|CA)\cap\mathcal{P}(C|AB)
\end{align}
\end{subequations}
In the present work, we aim to explore the set inclusion relations among these different convex sets.    

\section{Results}\label{sec3}
Let us first recall some already known structures. It seems tempting to assume that the set $\mathcal{B}^{int}(ABC)$ should be identical to the set of fully separable states $\mathcal{F}(ABC)$. Quite surprisingly, even for the simplest case of $3$-qubit system, this intuition is not correct. It turns out that $\mathcal{F}(2\otimes2\otimes2)$ is a strict subset of $\mathcal{B}^{int}(2\otimes2\otimes2)$. Example of a state belonging in $\mathcal{B}^{int}(2\otimes2\otimes2)$ but not belonging in $\mathcal{F}(2\otimes2\otimes2)$ follows from the construction of unextendible product basis (UPB) in $(\mathbb{C}^2)^{\otimes3}$ \cite{Bennett99}: 
\begin{align}\label{ShiftsUPB}
\rotatebox[origin=c]{0}{$\mathcal{U}^{{\bf Shifts}}_{PB}\equiv$}
\left\{\!\begin{aligned}
\ket{S_1}&:=\ket{0,1,+},~~\ket{S_2}:=\ket{1,+,0}\\
\ket{S_3}&:=\ket{+,0,1},~~\ket{S_4}:=\ket{-,-,-}
\end{aligned}\right\},	
\end{align}
where $\ket{\pm}:=\frac{1}{\sqrt{2}}(\ket{0}\pm\ket{1})$ and $\ket{x,y,z}\in(\mathbb{C}^2)^{\otimes3}$ stands as a short hand notation for $\ket{x}_A\otimes\ket{y}_B\otimes\ket{z}_C$. Let us consider the $3$-qubit state
\begin{align}\label{rhosupb}
\rho_{SU}:=\frac{1}{4}\left(\mathbf{I}_8-\sum_{i=1}^4\ket{S_i}\bra{S_i}\right).
\end{align}
From the property of UPB it follows that $\rho_{SU}\not\in\mathcal{F}(2\otimes2\otimes2)$. However, as shown in \cite{Bennett99} the state is biseparable across all the three bipartite cuts and hence $\rho_{SU}\in\mathcal{B}^{int}(2\otimes2\otimes2)$. In other words, being separable across all possible bipartitions, the state $\rho_{SU}$ contains multipartite entanglement.  

The authors in \cite{Bastian11,Leonardo13} find simple criteria to certify the presence of genuine entanglement in a multipartite state. To this aim the authors in \cite{Leonardo13} have made the conjecture that for $3$-qubit system the set of PPT-mixture will be identical to the set of biseparable states, {\it i.e.} $\mathcal{P}^{ch}(2\otimes2\otimes2)=\mathcal{B}^{ch}(2\otimes2\otimes2)$. However, Ha and Kye disprove this conjecture by constructing  $3$-qubit genuinely entangled states which are PPT \cite{Ha16}. Motivated by the results of Refs. \cite{Leonardo13} and \cite{Ha16}, here we address a different question. We ask whether a state being PPT across all three bipartitions implies that it is biseparable across all the three bipartite cuts. The question can be reformulated as to whether the set $\mathcal{B}^{int}(ABC)$ is same as the set $\mathcal{P}^{int}(ABC)$ or it is a proper subset. In the next sections, we answer this question in negation. 
\begin{theorem}\label{theo1}
$\mathcal{B}^{int}(ABC)\subsetneq\mathcal{P}^{int}(ABC)$ for $\mathbb{C}_A^{d_1}\otimes \mathbb{C}_B^{d_2}\otimes \mathbb{C}_C^{d_3}$ with $\min\{d_1,d_2,d_3\}\ge2$.
\end{theorem}
A state being separable in some bipartition must be PPT across that bipartition. It, therefore, follows that a state belonging to the set $\mathcal{B}^{int}(ABC)$ must also belong to the set $\mathcal{P}^{int}(ABC)$. To prove the strict inclusion relation we provide explicit construction of states $\rho_{ABC}$ that belong to $\mathcal{P}^{int}(ABC)$ but do not belong to $\mathcal{B}^{int}(ABC)$. Will will first discuss the construction for $d\otimes d\otimes d$ with $d\ge3$ and then discuss the construction in $2\otimes2\otimes2$. 
\subsection{Construction in $\mathbb{C}^d\otimes\mathbb{C}^d\otimes\mathbb{C}^d$ for $d\ge3$}
Consider the case $d=3$. The construction follows from a recently proposed UPB in $(\mathbf{C}^3)^{\otimes3}$ \cite{Agrawal19}. We will work with the computational basis $\left\{\ket{p,q,r}~|p,q,r=0,1,2\right\}$ for the Hilbert space $(\mathbf{C}^3)^{\otimes3}$. Consider the following {\it twisted} orthogonal product basis (t-OPB), 
\begin{subequations}\label{tCOPB}
\begin{align}
\mathbb{B}_0:&=\{\ket{\psi}_{kkk}\equiv\ket{k,k,k}~|~k\in\{0,1,2\}\},\\
\mathbb{B}_1:&=\{\ket{\psi(i,j)}_{1} \equiv \ket{0,\eta_i,\xi_j}\},\\
\mathbb{B}_2:&=\{\ket{\psi(i,j)}_{2} \equiv \ket{\eta_i,2,\xi_j}\},\\
\mathbb{B}_3:&=\{\ket{\psi(i,j)}_{3} \equiv \ket{2,\xi_j,\eta_i}\},\\
\mathbb{B}_4:&=\{\ket{\psi(i,j)}_{4} \equiv \ket{\eta_i,\xi_j,0}\},\\
\mathbb{B}_5:&=\{\ket{\psi(i,j)}_{5} \equiv \ket{\xi_j,0,\eta_i}\},\\
\mathbb{B}_6:&=\{\ket{\psi(i,j)}_{6} \equiv \ket{\xi_j,\eta_i,2}\},
\end{align}
\end{subequations}
where $i,j\in\{0,1\}$, and $\ket{\eta_i}:=\ket{0}+(-1)^i\ket{1}$, $\ket{\xi_j}:=\ket{1}+(-1)^j\ket{2}$. Consider the state $\ket{S}:=(\ket{0}+\ket{1}+\ket{2})^{\otimes 3}$. Note that $\ket{S}$ is not orthogonal to the states in $\mathbb{B}_0$ and also not orthogonal to the states $\{\ket{\psi(0,0)}_l\}_{l=1}^6$. With the remaining states it is orthogonal and accordingly the set of states 
\begin{align}\label{UPB3}
\mathcal{U}^{[3]}_{PB}:=\left\{\bigcup_{l=1}^6\left\{\mathcal{B}_l\setminus\ket{\psi(0,0)}_l\right\}\bigcup\ket{S}\right\} 
\end{align}
forms a UPB in $(\mathbf{C}^3)^{\otimes3}$ \cite{Agrawal19}. Cardinality of $\mathcal{U}^{[3]}_{PB}$ is $19$ and the $8$-dimensional subspace orthogonal to this UPB is a fully entangled subspace. The normalized projector on this fully entangled subspace is a rank-$8$ density operator $\rho^{[3]}(8)\in\mathcal{D}\left((\mathbf{C}^3)^{\otimes3}\right)$ given by
\begin{align}
\rho^{[3]}(8):=\frac{1}{8}\left(\mathbb{I}_{27}-\sum_{\ket{\psi}\in\mathcal{U}^{[3]}_{PB}}\ket{\tilde{\psi}}\bra{\tilde{\psi}}\right).
\end{align}
Here, $\ket{\tilde{x}}$ denotes the normalized states proportional to the unnormalized ray vector $\ket{x}$.
\begin{proposition}\label{prop1}
The state $\rho^{[3]}(8)$ is PPT in all three bipartitions but not separable in any of these bipartite cuts, {\it i.e.} $\rho^{[3]}(8)\in\mathcal{P}^{int}(3\otimes3\otimes3)$ but $\rho^{[3]}(8)\not\in\mathcal{B}^{int}(3\otimes3\otimes3)$.
\end{proposition}
\begin{proof}
Partial transposition (PT) acts linearly and on a product state $\rho_{AB}=\ket{\phi}_A\bra{\phi}\otimes\ket{\chi}_B\bra{\chi}$ it acts as $\rho^{\mathrm{T}_B}_{AB}:=(\mathbf{I}\otimes T)[\ket{\phi}_A\bra{\phi}\otimes\ket{\chi}_B\bra{\chi}]:=\ket{\phi}_A\bra{\phi}\otimes\ket{\chi^\star}_B\bra{\chi^\star}$ \cite{DiVincenzo03}; where $\ket{\chi^\star}=\sum\alpha_i^\star\ket{i}$ for $\ket{\chi}=\sum\alpha_i\ket{i}$ with $\{\ket{i}\}$ being an orthonormal basis (ONB). Therefore we have, $\rho^{\mathrm{T}_x}_3(8)\ge0$ for $x\in\{A,B,C\}$. In fact, since all the coefficients of the states in t-OPB $\{\mathbb{B}_l\}_{l=0}^6$ are real, therefore $\rho^{\mathrm{T}_x}_3(8)=\rho^{[3]}(8)$ for all $x$. This implies that  the state $\rho^{[3]}(8)$ is PPT across all bipartitions and consequently $\rho^{[3]}(8)\in\mathcal{P}^{int}(3\otimes3\otimes3)$. A state to belong in the set $\mathcal{B}(A|BC)$ must allow separable decomposition across this cut. But, as pointed out in \cite{Agrawal19}, one can have only four mutually orthogonal states that are separable across A|BC cut and orthogonal to the states in $\mathcal{U}^{[3]}_{PB}$\footnote{Using this construction, the authors in \cite{Agrawal19} have introduced the concept of unextendible biseparable basis (UBB). Construction of UBB is quite important as the subspace orthogonal to it turns out to be genuinely entangled subspace. This construction is also relevant to the study of genuine quantum nonlocality without entanglement (GQNWE) phenomenon \cite{Halder19(1),Rout19,Rout20}. GQNWE is a true multiparty generalization of the seminal quantum nonlocality without entanglement \cite{Bennett99(1)} phenomenon which has also been studied beyond quantum scenario very recently \cite{Bhattacharya20}.}. These four states are given by $\{\ket{\psi^-}_{24},\ket{\psi^-}_{56},\ket{\psi^-}_{(0)1},\ket{\psi^-}_{(2)3}\}$, where $\ket{\psi^-}_{lm}:=\ket{\psi(0,0)}_l-\ket{\psi(0,0)}_m$ and $\ket{\psi^-}_{(k)l}:=4\ket{k,k,k}-\ket{\psi(0,0)}_l$. Since the state $\rho^{[3]}(8)$ is of rank $8$ there is deficit of separable states across $A|BC$ cut and thus $\rho^{[3]}(8)$ does not allow a separable decomposition across this cut implying $\rho^{[3]}(8)\not\in\mathcal{B}(A|BC)$. From the symmetry of the construction it follows that $\rho^{[3]}(8)$ belongs to neither $\mathcal{B}(B|CA)$ nor $\mathcal{B}(C|AB)$ and consequently it follows that $\rho^{[3]}(8)\not\in\mathcal{B}^{int}(ABC)$. This completes the proof. 
\end{proof}
Proposition \ref{prop1} can be generalized for arbitrary higher dimensional Hilbert spaces $\left(\mathbb{C}^d\right)^{\otimes3}$ using the the UPBs constructed in \cite{Agrawal19}. For the general construction we refer the interested readers to \cite{Agrawal19}. Here we only recall the UPB of $\left(\mathbb{C}^4\right)^{\otimes3}$ since the construction for even dimension is bit difference than the odd dimensional case. The t-OPB in $\left(\mathbb{C}^4\right)^{\otimes3}$ is given by,
\begin{subequations}\label{tCOPB4}
\begin{align}
\mathbb{B}_0:&=\{\ket{\psi}_{kkk}\equiv\ket{k,k,k}~|~k\in\{0,3\}\},\\
\mathbb{B}^\prime_0:&=\{\ket{\psi(l,m,p)}\equiv\ket{\phi_l,\phi_m,\phi_p}\},\\
\mathbb{B}_1:&=\{\ket{\psi(i,j)}_{1} \equiv \ket{0,\eta_i,\xi_j}\},\\
\mathbb{B}_2:&=\{\ket{\psi(i,j)}_{2} \equiv \ket{\eta_i,3,\xi_j}\},\\
\mathbb{B}_3:&=\{\ket{\psi(i,j)}_{3} \equiv \ket{\xi_j,0,\eta_i}\},\\
\mathbb{B}_4:&=\{\ket{\psi(i,j)}_{4} \equiv \ket{\xi_j,\eta_i,3}\},\\
\mathbb{B}_5:&=\{\ket{\psi(i,j)}_{5} \equiv \ket{3,\xi_j,\eta_i}\},\\
\mathbb{B}_6:&=\{\ket{\psi(i,j)}_{6} \equiv \ket{\eta_i,\xi_j,0}\},
\end{align}
\end{subequations}
where $l,m,p\in\{0,1\},~\ket{\phi_0}:=\ket{1}+\ket{2},~\ket{\phi_1}:=\ket{1}-\ket{2};~i,j\in\{0,1,2\},~\ket{\eta_0}=\ket{0}+\ket{1}+\ket{2}$ and $\ket{\eta_1}~\&~\ket{\eta_2}$ are linear combination of $\{\ket{0},\ket{1},\ket{2}\}$ such that $\{\ket{\eta_i}\}_{i=0}^2$ are mutually orthogonal; $\ket{\xi_0}=\ket{1}+\ket{2}+\ket{3}$ and $\ket{\xi_1}~\&~\ket{\xi_2}$ are linear combination of $\{\ket{1},\ket{2},\ket{3}\}$ such that $\{\ket{\xi_i}\}_{i=0}^2$ are mutually orthogonal. Considering $\ket{S}:=(\ket{0}+\ket{1}+\ket{2}+\ket{3})^{\otimes3}$ the UPB is given by
\footnotesize
\begin{align}\label{UPB4}
\mathcal{U}^{[4]}_{PB}:=\left\{\bigcup_{l=1}^6\left\{\mathcal{B}_l\setminus\ket{\psi(0,0)}_l\right\}\bigcup\left\{\mathcal{B}^\prime_0\setminus\ket{\psi(0,0,0)}\right\}\bigcup\ket{S}\right\}. 
\end{align}
\normalsize
The rank-$8$ state $\rho^{[4]}(8)$ belonging to $\mathcal{P}^{int}(4\otimes4\otimes4)$ but not belonging to $\mathcal{B}^{int}(4\otimes4\otimes4)$ is given by
\begin{align}
\rho^{[4]}(8):=\frac{1}{8}\left(\mathbb{I}_{64}-\sum_{\ket{\psi}\in\mathcal{U}^{[4]}_{PB}}\ket{\tilde{\psi}}\bra{\tilde{\psi}}\right).
\end{align}
In the next subsection, we will discuss the construction in $3$-qubit Hilbert space. 

\subsection{Construction in $\mathbb{C}^2\otimes\mathbb{C}^2\otimes\mathbb{C}^2$}
We would first like to point out a fundamental difference between the $3$-qubit unextendible product basis $\mathcal{U}_{PB}^{{\bf Shifts}}$ and the UPBs $\mathcal{U}^{[3]}_{PB}$ \& $\mathcal{U}^{[4]}_{PB}$ (and their generalization). If biseparable states are allowed along with the states in $\mathcal{U}_{PB}^{{\bf Shifts}}$ then one can construct a complete orthogonal basis for $(\mathbb{C}^2)^{\otimes3}$. For instance, consider the two-qubit states $\ket{a}=\ket{1,+},~\ket{b}=\ket{+,0},~\ket{c}=\ket{0,1},~\ket{d}=\ket{-,-}$. Then the states 
\begin{align*}\label{A|BC}
\rotatebox[origin=c]{0}{$\$(A|BC)\equiv$}
\left\{\!\begin{aligned}
\ket{\kappa_1}&:=\ket{0}\ket{a^\perp},~~\ket{\kappa_2}:=\ket{1}\ket{b^\perp}\\
\ket{\kappa_3}&:=\ket{+}\ket{c^\perp},~~\ket{\kappa_4}:=\ket{-}\ket{d^\perp}
\end{aligned}\right\},	
\end{align*}
are separable across A|BC cut, where $\ket{a^\perp},\ket{b^\perp}\in\mbox{Span}\{\ket{a},\ket{b}\}$ and $\ket{c^\perp},\ket{d^\perp}\in\mbox{Span}\{\ket{c},\ket{d}\}$ with $\ket{x^\perp}$ denoting the state orthogonal to $\ket{x}$. The states in $\$(A|BC)$ along with the states in $\mathcal{U}_{PB}^{{\bf Shifts}}$ form a complete basis for $(\mathbb{C}^2)^{\otimes3}$. On the other hand, as discussed in the proof of Proposition \ref{prop1} the set $\mathcal{U}^{[3]}_{PB}$ (or its generalization $\mathcal{U}^{[d]}_{PB}$) is not completable by appending only biseparable states. This fact plays the crucial role to prove non biseparability of the the state $\rho^{[3]}(8)$ and its generalization. Such a fact, however, is not possible for $3$-qubit product basis as it is not possible to have a set of orthogonal product states in $\mathbb{C}^2\otimes\mathbb{C}^n$ which in uncompletable \cite{Bennett99,DiVincenzo03}. We therefore look for a different method to construct our required state in $3$-qubit Hilbert space. 

At this stage, we check the state of Ha and Kye \cite{Ha16} which is used to disprove the conjecture made by the authors in Ref.\cite{Leonardo13}. This particular state is PPT across A|BC and hence it belongs to $\mathcal{P}(A|BC)$. However, being NPT across other cuts it does not belong to $\mathcal{P}^{int}(2\otimes2\otimes2)$ and hence fails to fulfil our purpose. 

We then consider the PPT bound entangled state of $\mathbb{C}^2\otimes\mathbb{C}^4$ as constructed by P. Horodecki in \cite{Horodecki97}. This particular state can also be thought as a state of $\mathbb{C}^2\otimes\mathbb{C}^2\otimes\mathbb{C}^2$. Let us recall Horodecki's construction as of a $3$-qubit state. For that first consider the states,

\begin{subequations}

\begin{align}
\ket{\psi^1}_{ABC}=&\frac{1}{\sqrt{2}}(\ket{0}_A\otimes\ket{00}_{BC}+\ket{1}_A\otimes\ket{01}_{BC}),\\ 
\ket{\psi^2}_{ABC}=&\frac{1}{\sqrt{2}}(\ket{0}_A\otimes\ket{01}_{BC}+\ket{1}_A\otimes\ket{10}_{BC}),\\
\ket{\psi^3}_{ABC}=&\frac{1}{\sqrt{2}}(\ket{0}_A\otimes\ket{10}_{BC}+\ket{1}_A\otimes\ket{11}_{BC}),\\
\ket{\phi^{(b)}}_{ABC}&=\ket{1}_A\otimes\left(\sqrt{\frac{1+b}{2}}\ket{00}+\sqrt{\frac{1-b}{2}}\ket{11}\right)_{BC},
\end{align}
\end{subequations}
where $0\le b\le1$. Consider the density operator defined by,
\begin{align}
\chi_{ABC}:=\frac{2}{7}\sum_{i=1}^3\mathbb{P}\left[\psi^i_{ABC}\right]+\frac{1}{7}\mathbb{P}[011],  
\end{align}

where $\mathbb{P}[x]:=\ket{x}\bra{x}$. A straightforward calculation yields that the state $\chi_{ABC}$ is NPT across A|BC cut. Consider a new density operator,
\begin{align}\label{horo}
\sigma^{(b)}_{ABC}:=\frac{7b}{7b+1}\chi_{ABC}+\frac{1}{7b+1}\mathbb{P}\left[\phi^{(b)}_{ABC}\right].    
\end{align}
The state $\sigma^{(b)}_{ABC}$ turns out to be PPT across A|BC cut for the whole range of the parameter $b$. Here, the state $\mathbb{P}\left[\phi^{(b)}_{ABC}\right]$ can be thought as the noise part that absorbs the NPT-ness of $\chi_{ABC}$. Furthermore, applying the range criterion of entanglement (Theorem 2 of \cite{Horodecki97}) it turns out that for $0<b<1$ the state is entangled whereas it is separable for $b=0,1$. Therefore, for $b\in(0,1)$ the state $\sigma^{(b)}_{ABC}\in\mathcal{P}(A|BC)$ but $\sigma^{(b)}_{ABC}\not\in\mathcal{B}(A|BC)$. The matrix representation of $\sigma^{(b)}_{ABC}$ in computational basis reads as, 
\begin{align}
\sigma^{(b)}_{ABC}\equiv\frac{1}{7b+1} \begin{pmatrix}
b&0&0&0&0&b&0&0\\
0&b&0&0&0&0&b&0\\
0&0&b&0&0&0&0&b\\
0&0&0&b&0&0&0&0\\
0&0&0&0&\frac{1+b}{2}&0&0&\frac{\sqrt{1-b^2}}{2}\\
b&0&0&0&0&b&0&0\\
0&b&0&0&0&0&b&0\\
0&0&b&0&\frac{\sqrt{1-b^2}}{2}&0&0&\frac{1+b}{2}
\end{pmatrix}.
\end{align}
Here we follow the lexicographic order from left to right and from up to down. Explicit calculation further yields that $\sigma^{(b)}_{ABC}$ is NPT across the other two cuts for $b\in[0,1]$ and hence it does not belong to $\mathcal{P}^{int}(2\otimes2\otimes2)$. It is important to note that the construction of $\sigma^{(b)}_{ABC}$ is not party symmetric. So we consider a party symmetric state $\eta^{(b)}_{ABC}$ given by 
\begin{align}
\eta^{(b)}_{ABC}:=\frac{1}{3}\left(\sigma^{(b)}_{\overline{ABC}} + \sigma^{(b)}_{\overline{BCA}} + \sigma^{(b)}_{\overline{CAB}}\right).
\end{align}
We use the symbol like $\sigma^{(b)}_{\overline{ABC}}$ to denote the fact that the ordering of party index does matter. For instance, the state $\sigma^{(b)}_{\overline{ABC}}$ is same as $\sigma^{(b)}_{ABC}$ of Eq.(\ref{horo}), whereas $\sigma^{(b)}_{\overline{BCA}}$ is same as the state in Eq.(\ref{horo}) but with the role of the party indices changed as $A\to B,~ B\to C,~ C\to A$; the state $\sigma^{(b)}_{\overline{CAB}}$ is defined similarly. 
 
For certain range of the parameter $b$ the state $\eta^{(b)}_{ABC}$ turns out to be PPT across all the three bipartitions. But, we also find that the rank of this state as well as the rank of its partial transposition across different cuts turn out to be $8$. Therefore, the range criterion does not directly apply to establish inseparability of this state. To obtain a low rank density operator we thus consider the following operator, 
\begin{align}
h^{(b)}_{ABC}&:=\eta^{(b)}_{ABC}- \mu\left(v_1 v_1^{\mathrm{T}}+v_2 v_2^{\mathrm{T}}\right)+\nu\left(v_3 v_4^{\mathrm{T}}+v_4 v_3^{\mathrm{T}}\right)\nonumber\\
&~~~~~~~~~~~~~~~~~~+\epsilon\left(v_5 v_6^{\mathrm{T}}+v_6 v_5^{\mathrm{T}}\right),   
\end{align}
where, $\mu:= \frac{b}{3(1+7b)},~\nu~=\frac{1+3b}{6+42b},~\epsilon:=\frac{2b}{3(1+7b)}$; $\mathrm{T}$ denotes matrix transposition and
\begin{align*}
v_1&:=(0~1~0~0~0~0~1~0)^{\mathrm{T}},~~ 
v_2:=(0~0~1~0~0~1~0~0)^{\mathrm{T}},\\
v_3&:=(0~1~0~0~0~0~0~0)^{\mathrm{T}},~~
v_4:=(0~0~1~0~0~0~0~0)^{\mathrm{T}},\\
v_5&:=(0~0~0~0~0~0~1~0)^{\mathrm{T}},~~
v_6:=(0~0~0~0~0~1~0~0)^{\mathrm{T}}.
\end{align*}
Although the matrix $h^{(b)}_{ABC}$ is positive semi-definite, its trace is not one. A proper normalization yields us the density operator
\begin{align}
\rho^{[2]}_{ABC}(b):= \frac{3+21b}{3+17b}h^{(b)}_{ABC}.
\end{align}
Straightforward calculations lead us to the following observations regarding the state $\rho^{[2]}_{ABC}(b)$:
\begin{itemize}
\item[O-$1$]: $\rho^{[2]}_{ABC}(b)$ is a valid density operator for $b\in[0,1]$, {\it i.e.} for all values of the parameter $b$, $\rho^{[2]}_{ABC}(b)\in\mathcal{D}(2\otimes2\otimes2)$. 
\item[O-$2$]: Partial transposition of $\rho^{[2]}_{ABC}(b)$ with respect to A is positive semi-definite for parameter values $b\in[0,1]$, {\it i.e.} $\left[\rho^{[2]}_{ABC}(b)\right]^{\mathrm{T}_A}\ge0$ and consequently $\rho^{[2]}_{ABC}(b)\in\mathcal{P}(A|BC)$ for all $b\in[0,1]$.
\item[O-$3$]: Partial transpositions of $\rho^{[2]}_{ABC}(b)$ with respect to B and C are positive semi-definite for parameter values $b\in(\sim0.8184,1]:=\mathfrak{R}\subset[0,1]$, {\it i.e.} $\left[\rho^{[2]}_{ABC}(b)\right]^{\mathrm{T}_x}\ge0$ for $x\in\{B,C\}$ and consequently $\rho^{[2]}_{ABC}(b)\in\mathcal{P}(B|CA)$ and $\rho^{[2]}_{ABC}(b)\in\mathcal{P}(C|AB)$ for all $b\in\mathfrak{R}$.
\end{itemize}
We thus arrive at the following proposition.

\begin{proposition}\label{prop2}
For the parameter values $b\in\mathfrak{R}$ the state  $\rho^{[2]}_{ABC}(b)\in\mathcal{P}^{int}(2\otimes2\otimes2)$ but $\rho^{[2]}_{ABC}(b)\not\in\mathcal{B}^{int}(2\otimes2\otimes2)$ and hence $\mathcal{B}^{int}(2\otimes2\otimes2)\subsetneq\mathcal{P}^{int}(2\otimes2\otimes2)$.
\end{proposition}
\begin{proof}
Proof of the first part follows immediately from observations O-$2$ and O-$3$. Since $\rho^{[2]}_{ABC}(b)\in\mathcal{P}(A|BC)$ for $b\in[0,1]$ and $\rho^{[2]}_{ABC}(b)\in\mathcal{P}(B|CA),\mathcal{P}(C|AB)$ for $b\in\mathfrak{R}$, therefore $\rho^{[2]}_{ABC}(b)\in\mathcal{P}^{int}(2\otimes2\otimes2)$ for $b\in\mathfrak{R}$.

We will now prove that the state $\rho^{[2]}_{ABC}(b)$ is not separable across AB|C cut. For that we first write down the matrices of the density operator $\rho^{[2]}_{ABC}(b)$ and its partial transposition with respect to C, 
\begin{align}
\rho^{[2]}_{ABC}(b) \equiv\Theta \begin{pmatrix}
\Gamma & 0 & 0 &\frac{\Gamma }{3}& 0&\frac{\Gamma }{3}&\frac{\Gamma }{3}&0 \\
0&\Lambda&\Lambda&0&0&0&0&\Omega\\
0&\Lambda&\Lambda&0&0&0&0&\Omega\\
\frac{\Gamma }{3}&0&0&\Gamma&\frac{\Gamma }{3}&0&0&0\\
0&0&0&\frac{\Gamma }{3}&\Delta&0&0&\Omega\\
\frac{\Gamma }{3}&0&0&0&0&\frac{2\Gamma }{3}&\frac{2 \Gamma }{3}&0\\
\frac{\Gamma }{3}&0&0&0&0&\frac{2\Gamma }{3}&\frac{2\Gamma }{3}&0\\
0&\Omega&\Omega&0&\Omega&0&0&\zeta
\end{pmatrix};
\end{align}
\begin{align}
\left[\rho^{[2]}_{ABC}(b)\right]^{\mathrm{T}_c}\equiv\Theta \begin{pmatrix}
\Gamma & 0 & 0 &\Lambda& 0&0&\frac{\Gamma }{3}&0 \\
0&\Lambda&\frac{\Gamma }{3}&0&\frac{\Gamma }{3}&0&0&\Omega\\
0&\frac{\Gamma }{3}&\Lambda&0&0&\frac{\Gamma }{3}&0&0\\
\Lambda&0&0&\Gamma&0&0&\Omega&0\\
0&\frac{\Gamma }{3}&0&0&\Delta&0&0&\frac{2\Gamma }{3}\\
0&0&\frac{\Gamma }{3}&0&0&\frac{2\Gamma }{3}&\Omega&0\\
\frac{\Gamma }{3}&0&0&\Omega&0&\Omega&\frac{2\Gamma }{3}&0\\
0&\Omega&0&0&\frac{2 \Gamma }{3}&0&0&\zeta
\end{pmatrix};
\end{align}
where, $\Gamma:=\frac{b}{1+7b},~\Lambda:=\frac{1+3b}{6(1+7b)},~\Delta:=\frac{1+5b}{6(1+7b)},~\zeta:=\frac{1+b}{2(1+7b)},~\Omega:=\frac{2b+\sqrt{1-b^2}}{6(1+7b)}~\&~\Theta:=\frac{3+21b}{3+17b}$. A vector $\omega$ lying in the range of $\rho^{[2]}_{ABC}(b)$ can be expressed as,
\begin{align}
\label{range general}
\omega&=(A,B,B,C;D,E,E,F)^{\mathrm{T}},\\
&\mbox{where~}A,B,C,D,E,F\in \mathbb{C}.\nonumber
\end{align}
\begin{figure}[t!]
\centering
\includegraphics[width=8.5cm]{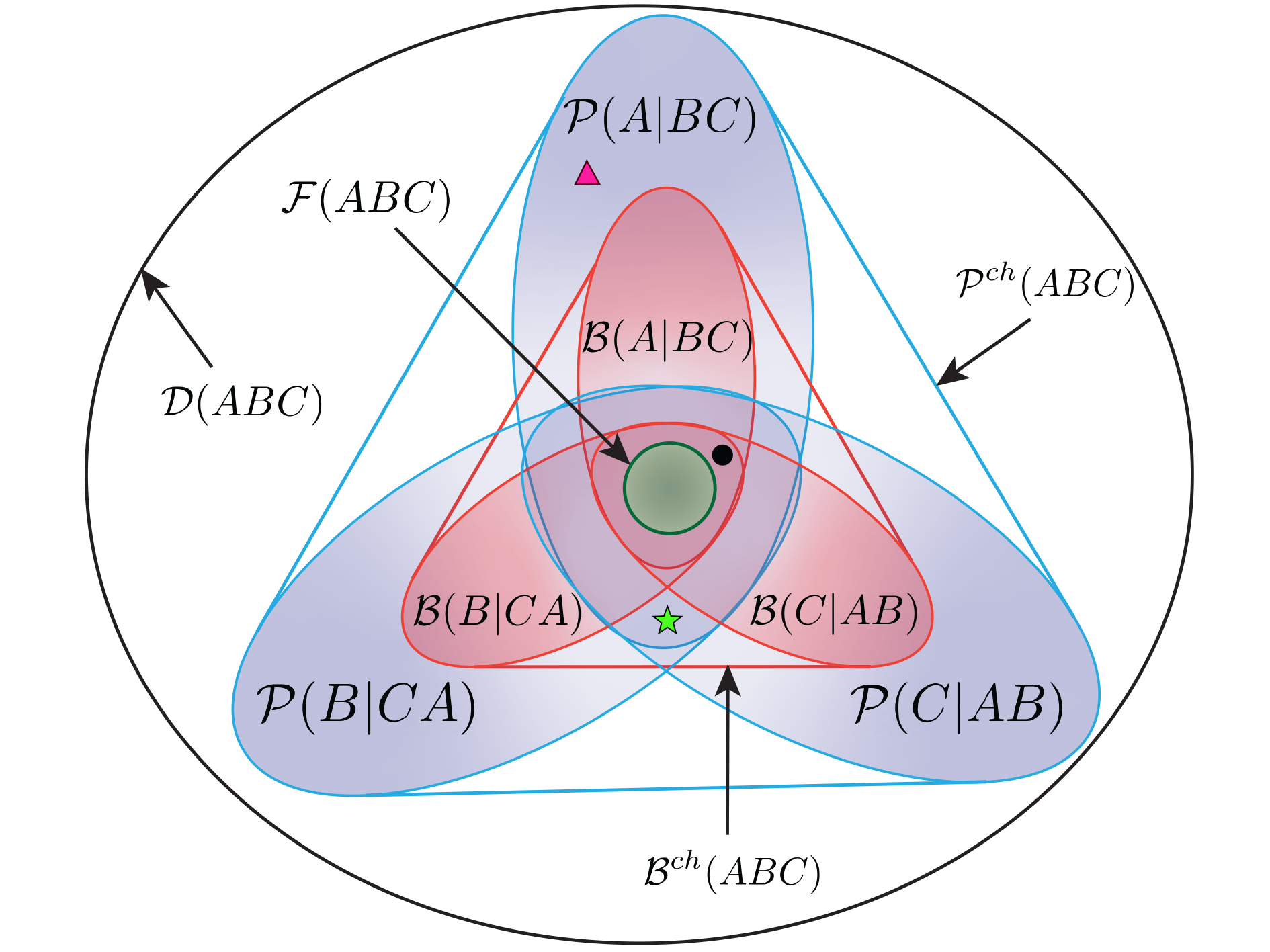}
\caption{(Color online) A set inclusion diagram among the different convex sets of states for tripartite quantum systems. For instance, $\mathcal{B}(A|BC)$ denotes the set of states that are separable across A|BC cut, whereas $\mathcal{P}(A|BC)$ stands for the set of PPT states across the same cut. The set of fully separable states $\mathcal{F}(ABC)$ (green region) is a strict subset of $\mathcal{B}(A|BC)\cap\mathcal{B}(B|CA)\cap\mathcal{B}(C|AB)$ even for the $3$-qubit Hilbert space. The point \textcolor[rgb]{0,0,0}{$\bullet$} representing the state $\rho_{SU}$ of Eq.(\ref{rhosupb}) (first identified in Bennett {\it et al}'s work \cite{Bennett99}) establishes this strict inclusion relation. On the other hand, even for the simplest tripartite system $\mathcal{B}^{ch}(2\otimes2\otimes2) \subsetneq \mathcal{P}^{ch}(2\otimes2\otimes2)$. The point 
\textcolor[rgb]{1,0,.5}{$\blacktriangle$} representing the state identified by Ha \& Kay in \cite{Ha16} establishes this particular strict inclusion relation. For any tripartite Hilbert space $\mathbb{C}^{d_1}_A\otimes\mathbb{C}^{d_2}_B\otimes\mathbb{C}^{d_3}_C$ with $\min\{d_1,d_2,d_3\}\ge2$, $\mathcal{B}(A|BC)\cap\mathcal{B}(B|CA)\cap\mathcal{B}(C|AB)\subsetneq\mathcal{P}(A|BC)\cap\mathcal{P}(B|CA)\cap\mathcal{P}(C|AB)$. The point \textcolor[rgb]{0,.8,0}{$\bigstar$} representing the states described in Propositions \ref{prop1} \& \ref{prop2} and in Remark \ref{remark1} establishes this fact.}
\label{figure1}
\end{figure}
If $\rho^{[2]}_{ABC}(b)$ is separable across AB|C cut, then according to the {\it range criterion} \cite{Horodecki97} there exists a set of product vectors $\{\psi_i \otimes \phi_k\}$ spanning the range space of $\rho^{[2]}_{ABC}(b)$ such that \{$\psi_i \otimes \phi_k^*\}$ span the range space of $\left[\rho^{[2]}_{ABC}(b)\right]^{\mathrm{T}_C}$ or any of the vectors \{$\psi_i \otimes \phi_k\}~\left(\{\psi_i \otimes \phi_k^*\}\right)$
belongs to the range of $\rho^{[2]}_{ABC}(b)~\left(\left[\rho^{[2]}_{ABC}(b)\right]^{\mathrm{T}_C}\right)$, with $\psi_i \in \mathbb{C}^4$ and $\phi_k \in \mathbb{C}^2$. Without any loss of generality, the elements of the set  \{$\psi_i \otimes \phi_k\}$ can be written as:
\begin{subequations}\label{r1}
 \begin{align}
\label{range vectors1}  v_1&=(\alpha,\beta,\gamma,\delta)^{\mathrm{T}}\otimes(1,0)^{\mathrm{T}},\\
\label{range vectors2}  v_2&=(\alpha,\beta,\gamma,\delta)^{\mathrm{T}}\otimes(0,1)^{\mathrm{T}}, \\
\label{range vectors3}  v_3&=(\alpha,\beta,\gamma,\delta)^{\mathrm{T}}\otimes(1,t)^{\mathrm{T}},
\end{align}
\end{subequations}
where $\alpha,\beta,\gamma,\delta,t \in \mathbb{C};~\&~ t\neq 0$. Comparing Eqs. (\ref{r1}) with Eq. (\ref{range general}) we obtain,   
\begin{subequations}\label{r2} 
\begin{align}
\label{range vectors final1}  v_1&=(A,0,D,0)^{\mathrm{T}}\otimes(1,0)^{\mathrm{T}},\\
\label{range vectors final2}  v_2&=(0,C,0,F)^{\mathrm{T}}\otimes(0,1)^{\mathrm{T}}, \\
\label{range vectors final3}  v_3&=(A,tA,D,tD)^{\mathrm{T}}\otimes(1,t)^{\mathrm{T}}.
\end{align}
\end{subequations}
Therefore, the partial complex conjugations of vectors in Eq.\ref{r2} are obtained as,
\begin{subequations}
\label{partial conjugate}
\begin{align}
\label{range conjugate1}  v_1^*&=(A,0,D,0)^{\mathrm{T}}\otimes(1,0)^{\mathrm{T}},\\
\label{range conjugate2}  v_2^*&=(0,C,0,F)^{\mathrm{T}}\otimes(0,1)^{\mathrm{T}}, \\
\label{range conjugate3}  v_3^*&=(A,tA,D,tD)^{\mathrm{T}}\otimes(1,t^*)^{\mathrm{T}}
\end{align}
\end{subequations}.
 These vectors should span range of $\left[\rho^{[2]}_{ABC}(b)\right]^{\mathrm{T}_C}$. Now, consider the vector $u=(0,0,\frac{b}{3+17b},0,0,\frac{2b}{3+17b},\frac{2b+\sqrt{1-b^2}}{6+34b},0)^{\mathrm{T}}$ in the range of $\left[\rho^{[2]}_{ABC}(b)\right]^{\mathrm{T}_C}$. This particular vector cannot be spanned by $\left\{v_1^*,v_2^*,v_3^*\right\}$ and hence leads to the fact that $\rho^{[2]}_{ABC}(b)$ is inseparable across AB|C cut. Thus, $\rho^{[2]}_{ABC}(b)\not\in\mathcal{B}(C|AB)$ and hence $\rho^{[2]}_{ABC}(b)\not\in\mathcal{B}^{int}(2\otimes2\otimes2)$ for $b\in(\sim0.8184,1]$. This completes the proof. 
\end{proof}
\begin{remark}\label{remark1}
Proposition \ref{prop1} and Proposition \ref{prop2} yield the proof for Theorem \ref{theo1} for $(\mathbb{C}^d)^{\otimes3}$ with $d\ge3$ and $(\mathbb{C}^2)^{\otimes3}$, respectively. Given these two Propositions it is also not hard to see that Theorem \ref{theo1} also holds for any tripartite Hilbert space $\mathbb{C}^{d_1}\otimes\mathbb{C}^{d_2}\otimes\mathbb{C}^{d_3}$ with $\min\{d_1,d_2,d_3\}\ge2$. For arbitrary $d_1,d_2,d_3$, construct the state $\rho^{[d_m]}_{ABC}$ either as of Proposition \ref{prop1} if $d_m:=\min\{d_1,d_2,d_3\}\ge3$ or as of Proposition \ref{prop2} if $d_m=2$. Clearly, $\rho^{[d_m]}_{ABC}\in\mathcal{P}^{int}(d_1\otimes d_2\otimes d_3)$ but $\rho^{[d_m]}_{ABC}\not\in\mathcal{B}^{int}(d_1\otimes d_2\otimes d_3)$.
\end{remark}
Theorem \ref{theo1} reveals nontrivial geometric implication regarding the state space structure of tripartite Hilbert spaces by established proper set inclusion relations among different convex sets of states (see Figure \ref{figure1}).

\section{Concluding remarks and future outlook}\label{sec4}
We have studied the intricate state space structure of multipartite quantum systems. In particular, we have shown that the intersection of three sets of biseparable states (across three different bipartite cuts) is a strict subset of the intersection of three sets of PPT states for tripartite Hilbert spaces $\mathbb{C}^{d_1}_A\otimes\mathbb{C}^{d_2}_B\otimes\mathbb{C}^{d_3}_C$ with  $\min\{d_1,d_2,d_3\}\ge2$. We establish this strict set inclusion relation by explicit construction of states that belongs to the set $\mathcal{P}^{int}(d_1\otimes d_2\otimes d_3)$ but not to $\mathcal{B}^{int}(d_1\otimes d_2\otimes d_3)$. At this point, the work by Eggeling \& Werner \cite{Eggeling01} is worth mentioning. There the authors studied state space structure for tripartite quantum systems by considering a particular class of states that commute with unitaries of the form $U\otimes U\otimes U$. While for three qubit system it turns out that $\mathcal{P}^{int}(2\otimes 2\otimes 2)=\mathcal{B}^{int}(2\otimes 2\otimes 2)$ if we limit within the $U\otimes U\otimes U$ invariant class, our results show that this is not the case for generic state space. Present study thus  provides new understanding towards the multipartite state space structure as well as the multipatite entanglement behavior and adds to the previous results established in \cite{Eggeling01,Ha16,Bennett99}.

Our work welcomes further questions for future study. First of all, we have only shown that the convex sets of states $\mathcal{P}^{int}(ABC)$ and $\mathcal{B}^{int}(ABC)$ are not identical. Now according to the classic Minkowski-Krein–Milman theorem, we know that every convex (and  compact) set in a Euclidean space (or more generally in a locally convex topological vector space) is the convex hull of its extreme points \cite{Krein40} (see also \cite{Zyczkowski06}). Since $\mathcal{P}^{int}(ABC)\neq\mathcal{B}^{int}(ABC)$, the sets of extreme point of these sets are also different, {\it i.e.}  $\mathcal{E}_{\mathcal{P}^{int}}(ABC)\neq\mathcal{E}_{\mathcal{B}^{int}}(ABC)$. Characterizing these extreme points will provide us a more detailed picture regarding the tripartite state space structure. In this respect, the work of \cite{Halder19} is worth mentioning, where the authors have shown that a $(d-3)/2$-simplex is sitting on the boundary between the set $\mathcal{P}(AB)$ and the set of non-PPT states for the Hilbert space $\mathbb{C}^d_A\otimes\mathbb{C}^d_B$ for odd $d$ with $d\ge3$. Our study also motivates questions on quantum dynamics. For the bipartite case, researchers have identified entanglement breaking completely positive trace-preserving maps (channels) $\mathcal{N}$ such that $(\mathcal{I}\otimes\mathcal{N})[\rho_{AB}]\in\mathcal{S}(AB)~\forall~\rho_{AB}\in\mathcal{D}(AB)$ \cite{Ruskai02,Horodecki2003(1)}. Similarly, here one might be interested in the classes of channels that map any tripartite state to the sets $\mathcal{F}(ABC)/\mathcal{B}^{int}(ABC)/\mathcal{P}^{int}(ABC)$, {\it i.e.} $(\mathcal{I}\otimes\mathcal{I}\otimes\mathcal{N})[\rho_{ABC}]\in\mathcal{F}(ABC)/\mathcal{B}^{int}(ABC)/\mathcal{P}^{int}(ABC)~\forall~\rho_{ABC}\in\mathcal{D}(ABC)$. Finally, comprehending the state space structure for an arbitrary number of systems is far from complete. 

{\bf Acknowledgements:} We would like to thank Prof. R. F. Werner for pointing out the relevant reference \cite{Eggeling01}. MB would like to thank Arup Roy for fruitful discussions. MB acknowledges research grant through INSPIRE-faculty fellowship from the Department of Science and Technology, Government of India.

\end{document}